\documentclass[11pt]{article}
\usepackage{authblk}
\usepackage{amsmath, amssymb}
\usepackage{amsthm}
\usepackage{enumitem}
\usepackage{color}
\usepackage{algorithmicx}
\usepackage[noend]{algpseudocode}
\usepackage{graphicx,tikz,subfig}
\usetikzlibrary{calc,fit}

\usepackage[margin=1in]{geometry}

\usepackage{titlesec}
\titlespacing*{\section}{0pt}{11pt}{11pt}
\titlespacing*{\subsection}{0pt}{11pt}{11pt}

\newtheorem*{rep@theorem}{\rep@title}
\newcommand{\newreptheorem}[2]{%
\newenvironment{rep#1}[1]
{ \def\rep@title{#2 \ref{##1}} \begin{rep@theorem}} {\end{rep@theorem}} }

\newtheorem{theorem}{Theorem}
\newreptheorem{theorem}{Theorem}
\newtheorem{lemma}[theorem]{Lemma}

\newtheorem{proposition}[theorem]{Proposition}
\newtheorem{corollary}[theorem]{Corollary}

\theoremstyle{definition}

\newcommand{\cB}{\mathcal{B}}

\newcommand{\cW}{\mathcal{W}}
\newcommand{\cV}{\mathcal{V}}

\newcommand{\E}{{\bf E}}
\DeclareMathOperator{\poly}{poly}

\DeclareMathOperator{\Var}{var}
\DeclareMathOperator{\FWPIT}{FindWeightIT}
\DeclareMathOperator{\FITBD}{FindITorBD}

\begin{document}

\title{Algorithms for weighted independent transversals and strong colouring}
\author[1]{Alessandra Graf}
\author[2]{David G. Harris}
\author[1]{Penny Haxell\thanks{Partially supported by NSERC.}}
\affil[1]{Department of Combinatorics and Optimization, University of Waterloo, Waterloo, ON, Canada}
\affil[2]{Department of Computer Science, University of Maryland, College Park, MD, United States}

\maketitle

\begin{abstract}
An {\it independent transversal} (IT) in a graph with a given vertex partition is an independent set consisting of one vertex in each partition class. Several sufficient conditions are known for the existence of an IT in a given graph and vertex partition, which have been used over the years to solve many combinatorial problems. Some of these IT existence theorems have algorithmic proofs, but there remains a gap between the best existential bounds and the bounds obtainable by efficient algorithms. 

Recently, Graf and Haxell (2018) described  a new (deterministic) algorithm that asymptotically closes this gap, but there are limitations on its applicability. In this paper we develop a randomized  algorithm that is much more widely applicable, and demonstrate its use by giving efficient algorithms for two problems concerning the strong chromatic number of graphs.
\end{abstract}

This is an extended version of an article which appeared in the  ACM-SIAM Symposium on Discrete Algorithms (SODA) 2021.

\section{Introduction}\label{intro}

Let $G = (V(G), E(G))$ be a graph with a partition $\mathcal V$ of its vertices; the elements of $\mathcal V$ are non-empty subsets of $V(G)$, which we refer to as \emph{blocks}.   For a vertex $v \in V(G)$, we let $\mathcal V(v)$ denote the unique block $U \in \mathcal V$ with $v \in U$.  The \emph{minimum blocksize} $b^{\text{min}}(\mathcal V)$ (or just $b^{\min}$ if $\mathcal V$ is understood) is the minimum size of any block in $\mathcal V$.  We say that $\mathcal V$ is \emph{$b$-regular} if  every block $U$ has size exactly $b$.  

We let $n = |V(G)|$ denote the number of vertices in $G$.  The \emph{neighbourhood} $N(v)$ of a vertex $v$ is the set of vertices $u \in V(G)$ with $uv \in E(G)$. The \emph{maximum degree} $\Delta(G)$  is the maximum number of neighbours of any vertex; again, if $G$ is understood, we write simply $\Delta$.

An independent set $I$ of $G$ is called an {\it independent transversal} (IT) of $G$ with respect to $\mathcal V$ if $|I \cap U| = 1$ for all $U \in \mathcal V$; likewise, $I$ is a {\it partial independent transversal} (PIT) of $G$ with respect to $\mathcal V$ if $|I \cap U| \leq 1$ for all $U \in \mathcal V$.   Many combinatorial problems can be formulated in terms of ITs in graphs with respect to given vertex partitions (see e.g.~\cite{Haxell2011}). Various results give sufficient conditions for the existence of an IT (e.g.~\cite{Alon1988, Fellows1990, Haxell1995, Haxell2001, Alon2008, Bissacot2011}).  In particular, \cite{Haxell1995, Haxell2001} showed the following:
\begin{theorem}[\cite{Haxell1995,Haxell2001}]
\label{th1ex}
If $G$ has a vertex partition  with $b^{\min}\geq 2 \Delta (G)$, then an IT of $G$ exists. 
\end{theorem}

This bound is optimal, since \cite{Szabo2006} showed that blocks of size $2\Delta-1$ are not sufficient to guarantee the existence of an IT. 

There is another important extension involving weighted ITs in the setting of vertex-weighted graphs. Theorem~\ref{th1ex} (which merely shows the existence of an IT without regard to weight) is not sufficient for these applications. For a weight function $w:V(G)\to\mathbb R$ and a subset $U\subseteq V(G)$, we write $w(U)=\sum_{u\in U}w(u)$ and $w^{\max}(U) = \max_{u \in U} w(u)$. We also write $w(G) = w(V(G)) = \sum_{v \in V(G)} w(v)$.  We say $w$ is \emph{non-negative} if $w(v) \geq 0$ for all $v$.  

Aharoni, Berger, \& Ziv \cite{Aharoni2007} showed the following (in a different but equivalent formulation):
\begin{theorem}[\cite{Aharoni2007}]
\label{th2ex}
If $G$ has a $b$-regular vertex partition with $b \geq 2 \Delta (G)$, then for any weight function $w: V(G) \rightarrow \mathbb R$ there exists an IT $M$ of $G$ with $w(M) \geq w(G)/b$.
\end{theorem}

The proofs of Theorem~\ref{th1ex} and Theorem~\ref{th2ex} are not algorithmic. There are some algorithms to efficiently return an IT given a graph $G$ and vertex partition $\mathcal V$, some of which can handle weighted graphs (see~\cite{Annamalai2017, Harris2016, Harris2017}). These mostly rely on algorithmic versions of the Lov\'asz Local Lemma (LLL); they typically give an IT under more stringent conditions of the form $b^{\text{min}} \geq c \Delta$, where $c$ is a constant strictly larger than $2$. The algorithm of \cite{Harris2016} has the condition $b^{\text{min}} \geq 4 \Delta -1 $, which is the strongest known criterion of this form.  

Recently, Graf \& Haxell \cite{Graf2018} developed a new algorithm, called $\FITBD$~\cite{Graf2018}, to find either an IT in $G$ or a set of blocks with a small dominating set $D$ which has some additional properties. The algorithm uses ideas from the original proof of Theorem~\ref{th1ex} and modifications of several key notions (including ``lazy updates'') from Annamalai \cite{Annamalai2017}. 

 To describe $\FITBD$, we require a few definitions.  A vertex set $D$ {\it dominates} another vertex set $W$ in $G$ if for all $w\in W$, there exists $uw \in E(G)$ for some $u\in D$. (This is also known as {\it strong domination} or {\it total domination}, but it is the only notion of domination we need so we use the simpler term.) For a subset $\cB \subseteq  \mathcal V$ of the vertex partition, we write $V(\cB) = \bigcup_{U \in\cB} U$. A {\it constellation} $K$ for $\cB$ is a pair of disjoint vertex sets $K^{\text{centre}}, K^{\text{leaf}} \subseteq V(\cB)$ with the following properties:
\begin{itemize}
\item $|K^{\text{leaf}}| = |\cB| - 1$
\item $K^{\text{leaf}}$ is a PIT of $G$ with respect to $\cV$
\item Each vertex $v \in K^{\text{centre}}$ has no neighbours in $K^{\text{centre}}$ and at least one neighbour in $K^{\text{leaf}}$
\item Each vertex $v \in K^{\text{leaf}}$ has  exactly one neighbour in $K^{\text{centre}}$ and no neighbours in $K^{\text{leaf}}$
\end{itemize}

We write $V(K)$ for the vertex set $K^{\text{centre}} \cup K^{\text{leaf}}$. The induced graph on $V(K)$, denoted $G[V(K)]$,  is thus a collection of stars, with centres and leaves in $K^{\text{centre}}$ and $K^{\text{leaf}}$ respectively.  These stars are all non-degenerate in the sense that they have at least one leaf.

We state a (slightly simplified) summary of the algorithm $\FITBD$ as follows:
\begin{theorem}[\cite{Graf2018}]\label{FITBD} 
The algorithm $\FITBD$ takes as input a parameter $\epsilon \in (0,1)$ and a graph $G$ with vertex partition $\mathcal V$ and finds either:  
\begin{enumerate}
\item an IT in $G$, or
\item a non-empty set $\cB \subseteq \mathcal V$ and a vertex set $D \subseteq V(G)$ such that $D$ dominates $V(\cB)$ in $G$ and $|D|<(2+\epsilon)(|{\mathcal B}|-1)$. Moreover, there is a constellation $K$ for some $\cB_0\supseteq\cB$ with $V(K) \subseteq D$  and $|D\setminus V(K)|<\epsilon(|\cB|-1)$.
\end{enumerate}
If $\Delta(G)$  and $\epsilon$ are fixed, then the runtime is $\poly(n)$.
\end{theorem}

 It is easy to show that if $b^{\min}(\mathcal V)\ge2\Delta(G)+1$, then no vertex set  $V(\cB)$ is dominated by a set of size less than $(2+\tfrac{1}{\Delta(G)})(|\cB|-1)$. This leads to the following result:

\begin{corollary}[\cite{Graf2018}]
\label{maxdegalg}
The algorithm $\FITBD$ takes as input a graph $G$ and a vertex partition with $b^{\text{min}} \geq 2 \Delta + 1$ and returns an IT in $G$.   If $\Delta \leq O(1)$, the runtime is $\poly(n)$.
\end{corollary}
This asymptotically matches the bound of Theorem~\ref{th1ex}. Thus Theorem~\ref{FITBD} and Corollary~\ref{maxdegalg} offer the possibility of new algorithmic proofs; see \cite{Graf2018, GrafThesis} for further details and applications.

From a combinatorial point of view, the algorithm $\FITBD$ is nearly optimal. However, from an algorithmic point of view, it is limited by its dependence on $\Delta$ and $\epsilon$, making it efficient only when these parameters are constant.  The aim of this paper is to give a new (randomized) algorithm that overcomes this limitation, as well as extending to the setting of weighted ITs. Our main theorem is as follows.

\begin{theorem}\label{main}
There is a randomized algorithm which takes as inputs a parameter $\epsilon > 0$, a graph $G$ with a $b$-regular vertex partition where $b \geq (2 + \epsilon) \Delta$, and a weight function $w: V(G) \rightarrow \mathbb R$, and finds an IT in $G$ with weight at least $w(G)/b$. For fixed $\epsilon$, the expected runtime is $\poly(n)$.
\end{theorem}

If we disregard vertex weights, this gives the immediate corollary:
\begin{corollary}
\label{maxdegalg2}
There is a randomized algorithm which takes as inputs a parameter $\epsilon > 0$ and a graph $G$ with a vertex partition where $b^{\text{min}} \geq (2 + \epsilon) \Delta$, and finds an IT of $G$. For fixed $\epsilon$, the expected runtime is $\poly(n)$.
\end{corollary}

In particular, Corollary~\ref{maxdegalg2} has Corollary~\ref{maxdegalg} as a special case (for constant $\Delta$, we can take $\epsilon = \tfrac{1}{ 2 \Delta}$), and is also stronger than all the previous LLL-based results (since we can also take $\epsilon$ to be an arbitrary fixed constant and allow $\Delta$ to vary freely).

The overall construction has three phases. In the first phase, in Section~\ref{FWPIT}, we develop an algorithm $\FWPIT$ which is an initial attempt to achieve Theorem~\ref{main}. This uses a streamlined and algorithmic version of a construction of Aharoni, Berger, and Ziv \cite{Aharoni2007}, overcoming some technical challenges in the analysis stemming from the fact that constellations provided by $\FITBD$ are slightly ``defective'' and only approximately dominate parts of the graph $G$, as compared to the non-constructive combinatorial bounds. On its own, this algorithm has two severe limitations:  while its runtime is polynomial in $n$, it is exponential in both the blocksize $b$ and the number of bits of precision used to specify the weight function $w$.

In the second phase, discussed in Section~\ref{provemain}, we use an algorithmic version of the LLL from Moser \& Tardos \cite{moser-tardos} to sparsify the graph. Given a vertex partition with blocksize $b \geq (2+\epsilon) \Delta$, where $\epsilon$ is an arbitrary constant, this effectively reduces the blocksize $b$ and the degree $\Delta$ to constant values, at which point  $\FWPIT$ can be used.  Unfortunately, a number of error terms accumulate in this process, including concentration losses from the degree reduction and quantization errors  from the $\FWPIT$ algorithm. As a consequence, this only gives an IT of weight $(1-\lambda) w(G)/b$, where $\lambda$ is an arbitrarily small constant.

In the third phase, carried out in Section~\ref{sec:LP}, we overcome this limitation by ``oversampling'' the high-weight vertices, giving the final result of Theorem~\ref{main}. For maximum generality, we analyse this in terms of a linear programming (LP) formulation related to  a construction of \cite{Aharoni2007}. 

In Section~\ref{sec:restrict}, we  demonstrate the use of Theorem~\ref{main} by providing algorithms for finding ITs which avoid a given set of vertices $L$ as long as $|L| < b^{\text{min}}$.  Such ITs are used in a number of constructions \cite{Rabern2013, Lo2019}, many of which do not themselves overtly involve the use of weighted ITs.

In Section~\ref{SC}, we consider strong colouring of graphs. For a positive integer $k$, we say that a graph $G$ is {\it strongly $k$-colourable with respect to vertex partition $\mathcal V$} if there is a proper vertex colouring of $G$ with $k$ colours so that no two vertices in the same block receive the same colour.  The {\it strong chromatic number} of $G$, denoted $s\chi(G)$, is the minimum $k$ such that $G$ is strongly $k$-colourable with respect to every vertex partition of $V(G)$ into blocks of size $k$. 

This notion was introduced independently by Alon~\cite{Alon1988, Alon1992} and Fellows~\cite{Fellows1990} and has been widely studied~\cite{Fleischner1992, Molloy2002, Haxell2004, Axenovich2006, Aharoni2007, Loh2007, Haxell2008, Lo2019}. The best currently-known explicit bound for strong chromatic number in terms of maximum degree  is $s\chi(G)\leq3\Delta(G)-1$, proved in~\cite{Haxell2004}. (See also~\cite{Haxell2008} for an asymptotically better bound.)  It is conjectured (see e.g.~\cite{Aharoni2007,Szabo2006}) that the correct general bound is $s\chi(G)\leq2\Delta(G)$. There is a natural notion of {\it fractional} strong chromatic number (see Section~\ref{SC}), for which the corresponding fractional version of this conjecture was shown in~\cite{Aharoni2007}. 

Graf \& Haxell \cite{Graf2018} used $\FITBD$ to develop an algorithm for strong colouring with $3\Delta+1$ colours, but, as before, the
algorithm is efficient only when $\Delta$ is constant. In Section~\ref{SC} we use Theorem~\ref{main} for an efficient algorithm for strong
colouring with $(3+\epsilon)\Delta$ colours for any fixed $\epsilon$, with no restrictions on $\Delta$.  We also give an
algorithmic version of the fractional strong colouring result.

We remark that if our goal was solely to show Corollary~\ref{maxdegalg2}, then the first and third phase of the proof of Theorem~\ref{main} could be completely omitted, and the second phase could use a simpler version of the LLL. However, Corollary~\ref{maxdegalg2} is not enough for our applications such as strong colouring.

We also remark that the sparsification step (Phase 2) is the only part of the overall algorithm that requires randomization. It is possible to derandomize the Moser-Tardos algorithm in this setting, giving fully deterministic algorithms for weighted ITs. This requires significant technical analysis of the Moser-Tardos algorithm beyond the scope of this paper; see \cite{harris2} for further details.

\section{FindWeightIT} \label{FWPIT}
Our starting point is a procedure $\FWPIT$ to  find a weighted IT using  $\FITBD$ as a subroutine. This takes as input a graph $G$ with a $b$-regular vertex partition $\mathcal V$ and a weight function $w$ on $G$.  It is defined as follows:
\begin{algorithmic}[1]
\Function{FindWeightIT}{$G, \mathcal V, w$}
	\State Set $W := \{v \in V(G): w(v) =  w^{\text{max}}(\mathcal V(v)) \}$ and $\mathcal W  := \{ U \cap W : U \in \mathcal V \}$.  \label{Wi}
	\State Apply $\FITBD$ to graph $G[W]$,  vertex partition $\cW$ and parameter $\epsilon = \frac{1}{8 b^2}$. \label{ffb}
	\If{$\FITBD(G[W],\cW, \epsilon)$ returns an IT $M'$} \label{ifIT}
		\State \Return $M:=M'$ \label{ifITreturn}
	\Else{ $\FITBD(G[W],\cW, \epsilon)$ returns $\cB \subseteq \cW$ and $D \subseteq W$ containing constellation $K$} \label{ifBD}
		\ForAll {$v\in V(G)$} $w'(v):=w(v)-|N(v)\cap D|$ \label{w'}
		\EndFor
		\State Recursively call  $M := \FWPIT (G,\mathcal V,  w')$. \label{recur}
		\State Set $Y:=\{v\in K^{\text{leaf}} \cap V(\cB):  |N(v)\cap D|=1\}$. \label{Y}
		\While{there is some vertex $v \in Y \setminus M$ with $N(v) \cap M = \emptyset$} \label{whilevloop}
		\State Choose such a vertex $v$ arbitrarily. \label{vY'}
		\State Update $M \leftarrow (M\cup \{v\})\setminus (M \cap \mathcal V(v))$. \label{ifBDreturn}
		\EndWhile
\State \Return $M$ \label{return}
\EndIf
\EndFunction
\end{algorithmic}

For a weight function $w: V(G) \to \mathbb R$ on $G$, we define  $|w|=\sum_{v \in V(G)} |w(v)|$, where $|w(v)|$ is the absolute value of $w(v)$. The main result we will show for this algorithm is the following:

\begin{theorem}\label{weight}
For an integer-valued weight function $w$ and a $b$-regular vertex partition $\mathcal V$ with $b > 2 \Delta(G)$, the algorithm $\FWPIT( G, \mathcal V, w)$ returns an IT in $G$ of weight at least $w(G)/b$. For fixed $b$, its runtime is $\poly(n,|w| )$.
\end{theorem}

Before we prove Theorem~\ref{weight}, we note that a simple quantization step can extend $\FWPIT$ to handle real-valued weight functions, with a small loss in the weight of the resulting IT.
\begin{lemma}\label{lem}
  There is an algorithm that takes as inputs a parameter $\eta>0$, a graph $G$ with a non-negative weight function $w$ and a $b$-regular vertex partition where $b > 2 \Delta(G)$, and finds an IT in $G$ of weight at least $(1-\eta) \tfrac{w(G)}{b}$. For fixed $b$, the runtime is $\poly(n, 1/\eta)$. 
\end{lemma}
\begin{proof}[Proof (assuming Theorem~\ref{weight})]
If $w(G) = 0$, then $w(v) = 0$ for all vertices $v$ and so $w$ is integral; in this case we can apply Theorem~\ref{weight} directly. So suppose that $w(G) > 0$. Define $\alpha=\frac{n}{\eta w(G)}$, and define a new weight function $w' : V(G) \rightarrow \mathbb Z_{\geq 0}$ by $w'(v)=\lfloor \alpha w(v) \rfloor$ for each $v$.  Apply Theorem~\ref{weight} to $G$ and $w'$; since $|w'| \leq \alpha w(G) = n/\eta$,  the runtime is $\poly(n,1/\eta)$ for fixed $b$. This generates an IT $M$ of $G$ with $w'(M) \geq w'(G)/b$. Here
$$
w'(G) =\sum\limits_{v\in V(G)} \lfloor \alpha w(v) \rfloor > \sum\limits_{v\in V(G)}(\alpha w(v)-1) = \alpha w(G) - n.
$$

Therefore
\[
\sum_{v\in M}w(v)\geq \sum_{v\in M} \frac{w'(v)}{\alpha} =\frac{w'(M)}{\alpha} \geq \frac{w'(G)}{b\alpha} > \frac{w(G)}{b} - \frac{n}{b \alpha} = \frac{w(G)}{b} (1 - \eta). \qedhere
  \]
\end{proof}

In the remainder of the section, we will prove Theorem~\ref{weight}. 

When we call $\FWPIT( G, \mathcal V, w)$, it generates a series of recursive calls on the same graph $G$ and vertex partition $\mathcal V$, but different weight functions $w^{(i)}$, where $w^{(0)} = w$ and $w^{(i+1)}$ is obtained from $w^{(i)}$ according to Line~\ref{w'}, i.e. $w^{(i+1)} = (w^{(i)})'$. To simplify notation, we assume throughout that  we have fixed a graph $G$ and a $b$-regular vertex partition $\mathcal V$ where $b > 2 \Delta$. 

\begin{proposition}\label{index}
In line~\ref{w'} we have $w'(G) > w(G) -  (b-1) (1 + \tfrac{1}{16 b^2}) (|\cB| - 1)$.
\end{proposition}
\begin{proof}
In order to reach line~\ref{w'}, $\FITBD(G[W];\cW)$ must return a non-empty set $\cB\subseteq \mathcal W$ of blocks and vertex set $D$ dominating  $V(\cB)$ with $|D| < (2 + \epsilon)  (|\cB| - 1)$. We can compute $w'(G)$ as:
$$
w'(G) = \sum_{v \in V} w'(v) = \sum_{v \in V} \bigl( w(v) - |N(v) \cap D|  \bigr) = w(G) - \sum_{v \in V} |N(v) \cap D|.
$$

In turn, we bound this as $$
\sum_{v \in V} |N(v) \cap D| = \sum_{x \in D} |N(x)| \leq |D| \Delta < (2+\epsilon)(|\cB| - 1) \Delta.
$$

Since $\epsilon = \frac{1}{8 b^2}$ and $b \geq 2 \Delta + 1$, we thus have
\[
\sum_{v \in V} |N(v) \cap D| < (2 + \tfrac{1}{8 b^2}) ( |\cB| - 1)  ( \tfrac{b-1}{2} )  =  (b-1) (1 + \tfrac{1}{16 b^2}) ( |\cB| - 1).  \qedhere
\]
\end{proof}

\begin{lemma}\label{time}
For fixed $b$,  $\FWPIT$ terminates in time $\poly(n,|w|)$.
\end{lemma}
\begin{proof}
We will show algorithm termination using a potential function $\Phi$ on weights $w$, defined as $$
\Phi(w) = -w(G) + b \sum_{U \in \mathcal V} w^{\max}(U) = \sum_{v \in V} (w^{\max}( \mathcal V(v)) - w(v)) \geq 0
$$

In each iteration where $\FWPIT$ reaches line 8, we have
$$
\Phi(w) - \Phi(w') = b \sum_{U \in \mathcal V} (w^{\max}(U) - w'^{\max}(U)) - (w(G) - w'(G)).
$$

As $D$ dominates $V(\cB)$, each vertex $v\in V(\cB)$ has $w'(v)\le w(v)-1$. By definition of $W$, this implies that 
$ {w'}^{\text{max}}(U) < w^{\text{max}}(U)$ for $U \in \cB$. So $\sum_{U \in \mathcal V} (w^{\max}(U) - w'^{\max}(U)) \geq |\cB|$. Combined with Proposition~\ref{index}, this shows 
$$
\Phi(w) - \Phi(w') \geq b |\cB| - (b-1) (1 + \tfrac{1}{16 b^2}) ( |\cB| - 1)  = b +  \bigl( 1 - \frac{b-1}{16 b^2} \bigr) \bigl( |\cB| - 1 \bigr) \geq b.
$$

Thus, $\Phi( w^{(i+1)} ) \leq \Phi( w^{(i)} ) - b$ in each iteration $i$. Since $\Phi( w^{(i)} ) \geq 0$ always, this implies that the total number of recursive calls starting from $w$ is at most $\Phi(w)/b \leq 2 |w|$. 

Next, let us check that each subproblem on weight function $w^{(i)}$ runs in $\poly(n, |w|)$ time.  The entries of $w^{(i)}$ are changed from $w = w^{(0)}$ by at most $\Delta i$, and so arithmetic operations take $\poly(n, |w|)$ time. $\FITBD$ runs in $\poly(n)$ time since $\Delta < b/2$ and $\epsilon = \frac{1}{8 b^2}$ and $b$ is fixed.  Finally, each execution of line~\ref{ifBDreturn} moves a vertex of $Y$ into $M$, and the vertex that gets removed from that block was not in $Y$ because $Y$ is a PIT and so the block contains at most one element of $Y$. Thus each iteration of the loop increases $|Y\cap M|$ by one, so it terminates within $n$ iterations.
\end{proof}

\begin{lemma}\label{PIT}
The set $M$ returned by $\FWPIT$ is an IT of $G$ with respect to $\mathcal V$.
\end{lemma}
\begin{proof}
We show this by strong induction on the runtime of $\FWPIT$. If $\FITBD(G[W],\cW, \epsilon)$ returns an IT, then $M$ is defined in line~\ref{ifITreturn} to be this same IT. Since $G[W]$ is an induced subgraph and $\cW$ is the restriction of $\cV$ to $W$, this is also an IT of $G$ with respect to $\cV$. 

Otherwise, when $\FITBD(G[W],\cW,\epsilon)$ returns a set $\cB$ of blocks and a set $D$ of vertices, $\FWPIT$ is recursively applied to obtain a set $M$ (line~\ref{recur}). The runtime on this recursive subproblem is clearly less than the runtime of the overall algorithm itself. So by the induction hypothesis, $M$ is an IT with respect to $\mathcal V$.
    
The set $M$ remains a transversal throughout the loop at line~\ref{whilevloop} because line~\ref{ifBDreturn} adds a vertex $v\in Y\setminus M$ to $M$ and removes the vertex of $M$ in the same block as $v$. Also, $M$ remains an independent set since $N(v)\cap M=\emptyset$. Thus $M$ at the end is an IT of $G$ with respect to $\mathcal V$.
\end{proof}

Thus $\FWPIT$ terminates quickly and returns an IT. It remains to show that the resulting IT $M$ has high weight. We first show a few preliminary results.

\begin{proposition}\label{decwprop}
The value of $w'(M)$ does not decrease during any iteration of the loop at line~\ref{whilevloop}.
\end{proposition}
\begin{proof}
  Let $v$ be the vertex chosen in line~\ref{vY'}, let $U  = \mathcal V(v)$ be the block containing $v$ and let $a = w ^{\text{max}} (U)$. By the definition of $Y$ (line~\ref{Y}), we know that $v \in W$ and that $v$ has exactly one neighbour in $D$. Thus $w(v)= a$ and $w'(v)=w(v)-1= a - 1$. 

Line~\ref{ifBDreturn} updates the transversal $M$ by adding $v$  and removing the vertex $x$ currently in $M \cap U$.  If $x \in W$, then, since $U \cap W \in \mathcal B$ and $D$ dominates $V(\cB)$, this means that $x$ has at least one neighbour in $D$, which implies that $w'(x) \leq w(x) - 1 \leq a - 1$.   Otherwise, if $x \notin W$, then $w(x) < a$. Since $w$ is integer-valued, this implies that $w(x) \leq a - 1$, and so $w'(x) \leq w(x) \leq a - 1$.  
 
 In either case, $w'(x) \leq w'(v)$, and so replacing $x$ by $v$ in $M$ does not decrease $w'(M)$.
\end{proof}

\begin{proposition}\label{yprop}
If $\FWPIT$ reaches line~\ref{ifBD}, then the output $M$ of $\FWPIT$ satisfies 
$$ \sum\limits_{v \in M} |N(v) \cap D| > \left(1-\tfrac{1}{16 b}\right)(|\cB|-1).$$
\end{proposition}
\begin{proof}
Because of the termination condition of the loop at line~\ref{whilevloop}, each
vertex $v \in Y\setminus M$ has a neighbor $u$ in $M$. Since $Y$ is a PIT, we know that $u \in M \setminus Y$. So there are at least $|Y \setminus M|$ edges from $M \setminus Y$ to $Y \setminus M$. Since $Y \subseteq D$, this in turn shows that $\sum_{v \in M \setminus Y} |N(v) \cap D| \geq |Y \setminus M|$. Also,  since $D$ dominates $V(\cB)$ and $Y\subseteq V(\cB)$, any vertex $v \in M \cap Y$ has $|N(v) \cap D| \geq 1$, and hence $\sum_{v \in M \cap Y} |N(v) \cap D| \geq  |M \cap Y|$. Putting these bounds together, we have
$$\sum\limits_{v \in M} |N(v) \cap D|=\sum\limits_{v \in M \setminus Y} |N(v) \cap D| + \sum\limits_{v \in M \cap Y} |N(v) \cap D| \geq |Y \setminus M| + |Y \cap M| = |Y|.$$

To complete the proof, we will show that $|Y| > \left(1 - \tfrac{1}{16 b}\right) (|\cB| - 1)$. To see this, note that
$$|Y| \geq |K^{\text{leaf}}| - | K^{\text{leaf}} \setminus V(\cB) | - | \{ y \in K^{\text{leaf}}  : |N(y) \cap D| \neq 1 \} | .
$$

Here $K$ is a constellation for some $\cB_0 \supseteq \cB$ and thus $|K^{\text{leaf}}| = | \cB_0 |- 1$. Since $K^{\text{leaf}}$ is a PIT, we have $|K^{\text{leaf}} \setminus V(\cB)| \leq |\cB_0| - |\cB|$. Each vertex in $K^{\text{leaf}}$ has zero neighbours in $K^{\text{leaf}}$ and exactly one neighbour in $K^{\text{centre}}$. Thus, if $y\in K^{\text{leaf}}$ has more than one neighbour in $D$, then it has a neighbour in $J = D \setminus V(K)$. (Recall that $V(K) \subseteq D$.) Theorem~\ref{FITBD} ensures that $|J| < \epsilon(|\cB|-1)$ so there are fewer than $\epsilon \Delta (|\cB| - 1)$ edges from $J$ to $K^{\text{leaf}}$ and hence fewer than $\epsilon \Delta (|\cB| - 1)$ vertices $y \in K^{\text{leaf}}$ with $|N(y) \cap D| \neq 1$.  As $\epsilon = \tfrac{1}{8 b^2}$ and $\Delta < b/2$, we have
\[  |Y| > ( | \cB_0| - 1 ) - ( |\cB_0| - |\cB|) - (\tfrac{1}{8 b^2}) (b/2) ( | \cB| - 1) = (1 - \tfrac{1}{16 b}) (|\cB| - 1). \qedhere  \]
\end{proof}

We are now ready to prove that the IT returned by $\FWPIT$ has the desired weight.

\begin{lemma}\label{weightM}
For $M = \FWPIT( G, \mathcal V, w)$, we have $w(M)\ge w(G) / b$.
\end{lemma}
\begin{proof}[Proof:] 
We prove this by strong induction on the runtime of $\FWPIT$.  If $\FITBD$ returns an IT $M'$ on the vertex set $W$, then
$$
w(M') = \sum_{U \in \cV} w^{\max}(U) \geq \sum_{U \in \cV} \frac{1}{b} \sum_{v \in U} w(v) =  w(G)/b,
$$
and we are done.

Otherwise, suppose $\FITBD$ returns $\cB$ and $D$ (i.e. lines~\ref{ifBD}--\ref{ifBDreturn} are executed). By Lemma~\ref{PIT}, the recursive call  $\FWPIT( G, \mathcal V, w')$ returns an IT $M$ at line~\ref{recur}. By the induction hypothesis, it satisfies $w'(M)\ge w'(G)/b$.  By Proposition~\ref{decwprop}, the value $w'(M)$ does not  decrease during the loop at line~\ref{whilevloop}, so the final output $M$ also has $w'(M) \geq w'(G)/b$.  By Proposition~\ref{index}, we have $w'(G) > w(G) - (b-1) (1 + \frac{1}{16 b^2}) ( |\cB| - 1)$, and so
$$w'(M) > \frac{w(G) - (b-1) (1 + \frac{1}{16 b^2}) ( |\cB| - 1)}{b} = \frac{w(G)}{b} - \left(1 + \tfrac{1}{16 b^2}\right)\left(1-\tfrac{1}{b}\right)(|\cB| - 1).$$

By Proposition~\ref{yprop}, we have $w(M) - w'(M) = \sum_{v \in M} |N(v) \cap D| > \left(1 - \tfrac{1}{16 b}\right) (|\cB| - 1)$.  Overall, this gives
\begin{align*}
w(M) &= w'(M) + (w(M) - w'(M)) \\
&> \Bigl( \frac{w(G)}{b }- \left(1 + \tfrac{1}{16 b^2}\right)\left(1 - \tfrac{1}{b}\right) ( |\cB| - 1) \Bigr) + \left(1 - \tfrac{1}{16 b}\right) ( |\cB| - 1)  \\
&= \frac{w(G)}{b} + \frac{15 b^2 - b + 1}{16 b^3} (|\cB| - 1) \geq \frac{w(G)}{b}. \qedhere
\end{align*}
\end{proof}

Theorem~\ref{weight} and Lemma~\ref{lem} now follow from Lemmas~\ref{time},~\ref{PIT}, and~\ref{weightM}.

\section{Degree Reduction}\label{provemain}
The next step in the proof is to remove the condition that $b$ is constant. Our main tool for this is the LLL, in particular the LLL algorithm of Moser and Tardos~\cite{moser-tardos}. The basic idea is to use the LLL for a ``degree-splitting'': we reduce the degree, the blocksize, and the total vertex weight of $G$ by a factor of approximately half. By doing this repeatedly, we scale down the original graph to a graph with constant blocksize. At that point we use $\FWPIT$.

Let us begin by reviewing the algorithm of Moser and Tardos.

\begin{theorem}[\cite{moser-tardos}]\label{mt-thm}
 There is a randomized algorithm which takes as input a probability space $\Omega$ in $k$ independent variables $X_1, ..., X_k$ along with a collection of ``bad'' events $B_1, \dots, B_{\ell}$ in that space, wherein each $B_i$ is a Boolean function of a subset of the variables $\text{var}(B_i)$.

If $e p d \leq 1$, where $p = \max_{i} \Pr_{\Omega}(B_i)$  and $d = \max_{i}  | \{ j : \text{var} (B_j) \cap \text{var}(B_i) \neq \emptyset \} |$, then the algorithm has expected runtime polynomial in $k$ and $\ell$ and outputs a configuration $X = (X_1, \dots, X_k)$ such that all bad-events $B_i$ are false on $X$.
\end{theorem}

One additional feature of this algorithm is critical for our application to weighted ITs: the output state $X$ produced by the Moser-Tardos algorithm has a probability distribution with nice properties \cite{hss, harris-srin, Harris2016a}. One result of~\cite{harris-srin}, which we present in a simplified form, is the following.

\begin{theorem}[\cite{harris-srin}]\label{mt-thm2}
Suppose the  conditions of Theorem~\ref{mt-thm} are satisfied. Let $E$ be an event in the probability space $\Omega$ which is a Boolean function of a subset of variables $\Var(E)$, and let $r$ be the number of bad-events $B_i$ with $\Var(E) \cap \Var(B_i) \neq \emptyset$, i.e., $B_i$ can affect $E$. Then the probability that $E$ holds in the output configuration $X$ of the Moser-Tardos algorithm  is at most  $e^{e p r} \Pr_{\Omega}(E)$.
\end{theorem}

Using the Moser-Tardos algorithm, we get the following degree-splitting algorithm.
\begin{lemma}\label{sparsify-one-round}
  There is a randomized polynomial-time algorithm that takes as input a graph $G$ with a non-negative weight function $w$ and a $b$-regular vertex partition $\mathcal V$ where $b \geq 15000$.  It generates an induced subgraph $G'$ such that   \\
  \noindent (i) every block $U \in \mathcal V$ has exactly $b' := \lceil b / 2 \rceil$ vertices in $G'$, \\
  \noindent (ii)  $\Delta(G') \leq D/2 + 10 \sqrt{D \log D}$ where we define $D = \max\{ b/3, \Delta(G) \}$ \\
  \noindent (iii) $\tfrac{w(G')}{b'} \geq  (1 - 1/b) \tfrac{w(G)}{b}$.
\end{lemma}
\begin{proof}
We will use Theorem~\ref{mt-thm}, where the probability space $\Omega$ has a variable $X_U$ for each block $U \in \mathcal V$; the distribution of $X_U$ is to select a uniformly random subset $U' \subseteq U$ of size exactly $b'$. We will  set $G'$ to be the induced graph on vertex set $V' = \bigcup_{U \in \mathcal V} X_U$.  This clearly satisfies property (i).  

For each vertex $v$, we have a bad-event $B_v$ that $v$ has more than $s = D/2 + 10 \sqrt{D \log D}$ neighbours in $V'$. If all events $B_v$ are false, then property (ii) will hold.  Note that any variable $X_U$ affects an event $B_v$ only if $N(v) \cap U \neq \emptyset$; so, $X_U$ can affect at most $b \Delta(G)$ events.  

To calculate the parameters $p$ and $d$ of Theorem~\ref{mt-thm}, consider some vertex $v$ with neighbours $y_1, \dots, y_k$. The event $B_v$ is affected by the variable $X_{U_i}$ for each block $U_i = \mathcal V(y_i)$; each $X_{U_i}$ in turn affects at most $b \Delta(G)$ bad-events. In total, $B_v$ affects at most $b \Delta(G)^2 \leq 3 D^3$ bad-events.

We next calculate the probability of $B_v$. The degree of $v$ in $V'$ is the sum $Y = \sum_{j=1}^k  Y_j$, where $Y_j$ is the indicator that $y_j \in V'$. The random variables $Y_j$ are negatively correlated and $Y$ has expectation $\E[Y] \leq \tfrac{k b'}{b}  \leq \tfrac{D+ 1}{2}$.  Hoeffding's inequality applies to sums of negatively correlated random variables (see, e.g., \cite{farcomeni}), giving:
$$
\Pr( Y \geq s) \leq \Pr (Y \geq \E[Y] + (10 \sqrt{D \log D} - 1/2) ) \leq e^{-2 (10 \sqrt{D \log D} - 1/2)^2/k}, 
$$
and for $D \geq 5000$ and $k \leq \Delta(G) \leq D$, this is at most $D^{-100}$.   

Thus $p \leq D^{-100}$ and $d \leq 3 D^3$, and $e p d \leq 1$. The Moser-Tardos algorithm generates a configuration avoiding all bad-events $B_v$, and the  resulting graph $G'$ satisfies (i) and (ii).  It remains to analyse $w(G')$.   

By Theorem~\ref{mt-thm2}, for any block $U$ and fixed $b'$-element set $A \subseteq U$, the probability of  $X_U = A$ in the algorithm output is at most $e^{e p r}$ times its probability in the original probability space $\Omega$, where $r$ is the number of bad-events affected by $X_U$. The original sampling probability is $1/\binom{b}{b'}$ and we have already seen that $r \leq b \Delta(G) \leq 3 D^2$. So 
$$\Pr(X_U = A) \leq \frac{e^{e D^{-100} \cdot 3 D^2}}{\binom{b}{b'}} \leq \frac{e^{D^{-97}}}{\binom{b}{b'}}.
$$

Consider the random variable $L = w(G) - w(G')$. Since $w$ is non-negative, we have:
\begin{align*}
\E[L] &= \sum_{U \in \mathcal V} \sum_{\substack{A \subseteq U \\ |A| = b'}} \Pr( X_U = A ) w(U \setminus A)  \leq \frac{e^{D^{-97} }}{\binom{b}{b'}} \sum_{U \in \mathcal V} \sum_{\substack{A \subseteq U \\ |A| = b'}}  w(U \setminus A) \\
& = \frac{e^{D^{-97}}}{\binom{b}{b'}}  \sum_{U \in \mathcal V} w(U) \tbinom{b-1}{b'}  = (1 - b'/b) e^{D^{-97}}   w(G) .
\end{align*}
By Markov's inequality applied to the non-negative random variable $L$, therefore, the bound
\begin{equation}
\label{l-eqn}
L \leq w(G) (1 - b'/b) (1 + D^{-2})
\end{equation}
holds with probability at least $1 - \frac{e^{D^{-97}}}{1 + D^{-2}} \geq \Omega(D^{-2})$. We can repeatedly call the algorithm until we get a configuration satisfying Eq.~(\ref{l-eqn}). This takes $O(D^{2}) = \poly(n)$ repetitions on average, and each iteration has expected runtime $\poly(n)$. The resulting graph $G'$ then has
 \begin{align*}
w(G') &= w(G) - L \geq w(G) \bigl(1  -  (1 - b'/b) (1 + D^{-2}) \bigr) = w(G) (b'/b) \bigl( 1 + D^{-2} - D^{-2} b / b' \bigr)
\end{align*} 
and thus, since $b' \geq b/2$ and $D \geq b/3 \geq 5000$, we have 
 \[
 w(G')/b' \geq (1 - D^{-2}) w(G)/b  \geq (1-1/b) w(G)/b. \qedhere 
 \]
\end{proof}

\begin{lemma}\label{sparse-multi}
There is a randomized algorithm that takes as input  parameters $\epsilon, \lambda \in (0,1)$, a graph $G$ with a $b$-regular vertex partition where $b \geq (2+\epsilon) \Delta(G)$ and  a  non-negative weight function $w$. It generates an IT with  weight at least $\frac{w(G)}{b} ( 1 - \lambda )$. For fixed $\epsilon$ and $\lambda$, the expected runtime is $\poly(n)$.
\end{lemma}
\begin{proof}
If $b \leq \frac{10^{20}}{\epsilon^3 \lambda}$, then we can simply apply Lemma~\ref{lem} directly. So, let us assume that $b > \frac{10^{20}}{\epsilon^3 \lambda}$.  Our strategy will be to repeatedly apply Lemma~\ref{sparsify-one-round} for $t =  \bigl \lfloor \log_2 \tfrac{b \epsilon^3 \lambda}{10^{20}} \bigr \rfloor$ rounds. This generates  a series of induced graphs $G_i = G[V_i]$ for $i = 0, \dots, t$ where $V(G) = V_0 \supseteq V_1 \supseteq \dots \supseteq V_t$, along with corresponding vertex partitions $\mathcal V_i = \{ U \cap V_i \mid U \in \mathcal V \}$.  
At the end of this process, we finish by applying Lemma~\ref{lem} to the graph $G_t$ to get the desired independent transversal. 

To analyse this process, let us recursively define parameters $b_i, \delta_i$ as:
\begin{eqnarray*}
b_0 = b, & &\delta_0 = \frac{b}{2 + \epsilon}, \\
b_{i+1} = \left\lceil  b_i/ 2 \right\rceil, && \delta_{i+1} = \delta_i/2 + 10 \sqrt{\delta_i \log \delta_{i}} \qquad \text{for $i = 0, \dots, t-1$.}
\end{eqnarray*}

Note the following straightforward bounds for $i \leq t$:
\begin{equation}
\label{g0a}
b_i = \lceil 2^{-i} b \rceil \geq 2^{-t} b \geq \frac{10^{20}}{\epsilon^3 \lambda}, \qquad \text{and} \qquad \delta_i \geq 2^{-i} \delta_0  \geq 2^{-t} \delta_0 \geq \frac{10^{20}}{3 \epsilon^3 \lambda}.
\end{equation}

Each partition $\mathcal V_i$ during this process will be $b_i$-regular.  The precondition of Lemma~\ref{sparsify-one-round} at each round $i$, namely $b_i \geq 15000$,  follows immediately from Eq.~(\ref{g0a}).  To explain the role of the parameter $\delta_i$,  we show the following three bounds for $i \leq t$ by induction on $i$:
\begin{eqnarray}
\label{g1}
&\delta_i \leq \delta_0 2^{-i}  + (\delta_0 2^{-i})^{2/3} \\
\label{g3}
& b_i < 3 \delta_i \\
\label{g4}
& \Delta(G_i) \leq \delta_i.
\end{eqnarray}
The base case $i = 0$ is clear for all of them. For the induction step for Eq.~(\ref{g1}), let $x =\delta_0 2^{-i}$. Applying the induction hypothesis $\delta_i \leq x + x^{2/3}$ gives
\begin{align*}
\delta_{i+1} &= \delta_i/2 + 10 \sqrt{\delta_i \log \delta_{i}} \leq (x + x^{2/3})/2  + 10 \sqrt{(x + x^{2/3}) \log(x + x^{2/3})} .
\end{align*}
Since $x = \delta_0 2^{-i} \geq 10^{20}/3$, it can be easily checked this is at most $ x/2 + (x/2)^{2/3}$ as desired. Next, for Eq.~(\ref{g3}), we have
$$
b_{i+1} - 3 \delta_{i+1} \leq (b_i/2 + 1/2) - 3 (\delta_i/2 + 10 \sqrt{\delta_i \log \delta_i}) = ( b_i - 3 \delta_i)/2 + (1/2 - 30  \sqrt{\delta_i \log \delta_i}).
$$
By the induction hypothesis, the first term is negative. Since $\delta_i \geq 10^{20}/3$, the second term is also negative. Thus we maintain $b_i < 3 \delta_i$ for all $i$.  

Finally, for Eq.~(\ref{g4}), note that when applying Lemma~\ref{sparsify-one-round} at round $i$, we have  $\Delta(G_{i+1}) \leq D_i/2 + 10 \sqrt{D_i \log D_{i}}$ where $D_i = \max\{ b_i/3, \Delta(G_i) \}$.  By the induction hypothesis, we have $\Delta(G_i) \leq \delta_i$ and $b_i/3 < \delta_i$. Thus, $D_i \leq \delta_i$ and so  $\Delta(G_{i+1}) \leq \delta_i/2 + 10 \sqrt{\delta_i \log \delta_i} = \delta_{i+1}$.
 
Now, after applying Lemma~\ref{sparsify-one-round} in every round, the resulting weights $w(G_i)$  satisfy:
$$
\frac{w(G_{i+1})}{b_{i+1}} \geq \frac{w(G_i)}{b_i} (1 - 1/b_i)
$$
and using the identity $1 -x \geq e^{-2 x}$ for $x \leq 1/2$, this telescopes as:
$$
\frac{w(G_t)}{b_t} \geq \frac{w(G)}{b} \prod_{i=0}^{t-1} (1 - 1/b_i) \geq  \frac{w(G)}{b} \prod_{i=0}^{t-1} e^{-2/b_i} = \frac{w(G)}{b} e^{-2 \sum_{i=0}^{t-1} 1/b_i}.
$$
From Eq.~(\ref{g0a}), we see that $\sum_{i=0}^{t-1} 1/b_i \leq \sum_{i=0}^{t-1} 2^i / b \leq 2^t/b \leq \frac{\epsilon^3 \lambda}{10^{20}} \leq \lambda/20$. Thus,  
$$
\frac{w(G_t)}{b_t} \geq \frac{w(G)}{b} e^{-\lambda/10}.
$$

Finally, we check the preconditions of Lemma~\ref{lem} for the graph $G_t$, i.e., that $b_t > 2 \Delta(G_t)$ and $b_t \leq O(1)$. First, from Eq.~(\ref{g1}) and Eq.~(\ref{g0a}) we get $$
\delta_t \leq \delta_0 2^{-t} (1 + (\delta_0 2^{-t})^{-1/3}) \leq \delta_0 2^{-t} \Bigl( 1 + \bigl( \frac{10^{20}}{3 \epsilon^3 \lambda}\bigr)^{-1/3} \Bigr) \leq 2^{-t}  \delta_0 (1 + \epsilon/10^6).
$$
Hence,  using Eq.~(\ref{g4}), we have 
$$
b_t -2 \Delta(G_t) > b_t - 2 \delta_t \geq 2^{-t} b - 2 \cdot 2^{-t} \delta_0 (1 + \epsilon / 10^6) = 2^{-t} b \Bigl( 1 - \frac{ 2 (1 + \epsilon/10^6) }{2 + \epsilon} \Bigr) > 0.
$$

Also, by definition of $t$, we have $b_t = \lceil b 2^{-t} \rceil \leq 1 + 2 \cdot 10^{20} / \epsilon^3 \lambda$; since $\epsilon, \lambda$ are fixed, this is fixed as well. So Lemma~\ref{lem} on graph $G_t$ and parameter $\eta = \lambda/2$ produces an IT $M$ of weight
\[
w(M) \geq \frac{w(G_t)}{b_t} (1-\eta) \geq \frac{w(G)}{b} ( 1 - \lambda/2) e^{-\lambda/10} \geq \frac{w(G)}{b} ( 1-\lambda). \qedhere
\]
\end{proof}
\section{An LP for Weighted ITs} \label{sec:LP}
To finish the proof of Theorem~\ref{main}, we need to remove the term  $1 - \lambda$
from Lemma~\ref{sparse-multi}. For maximum generality, we use an LP formulation adapted from \cite{Aharoni2007}, which relaxes the condition $b \geq (2 + \epsilon)  \Delta$ to allow each vertex $v$ to be taken with a fractional multiplicity $\gamma_v \in [0,1]$.  Formally, given a graph $G$, vertex partition $\mathcal V$, weight function $w$, and a value $\delta \in \mathbb R_{\geq 0}$, we define $\mathcal P_{G, \delta}$ to be the following LP. (Here, $\delta$ plays the role of $\frac{1}{2+\epsilon}$.)

\begin{equation*}
\begin{array}{rcrclcl}
\displaystyle \max & \multicolumn{3}{l}{\sum\limits_{v\in V(G)}w(v)\gamma_v} \\ 
\textrm{subject to} & \sum\limits_{u\in N(v)}&\gamma_u & \le & \delta & & \forall v\in V(G)\\ 
&\sum\limits_{v \in U}&\gamma_v & = &1 & & \forall U \in \mathcal V \\
& 0\le & \gamma_v & \le&1 & & \forall v \in V(G). 
\end{array}
\end{equation*}
If the LP $\mathcal P_{G, \delta}$ is feasible, we let $\tau_{G, w, \delta}$ be the largest objective function value.  The next results show how to get a fractional version of Theorem~\ref{main} in terms of $\mathcal P_{G, \delta}$.

\begin{proposition}\label{ABZ-thm3}
There is a randomized algorithm which takes as input parameters $\delta, \lambda \in (0,1/2)$, a graph $G$ with vertex partition $\mathcal V$, a vector $\gamma$ in $\mathcal P_{G, \delta}$, and a non-negative weight function $w$ on $G$, and returns an IT in $G$ with weight at least $(1 - \lambda) \sum_{v \in V(G)} \gamma_v w(v)$. For fixed $\delta$ and $\lambda$, the expected runtime is $\poly(n)$.
\end{proposition}
\begin{proof}
 Let $\epsilon = 1/2 - \delta$. Form a new graph $G'$ by creating, for each vertex $v \in V(G)$, a group of $\lceil b \gamma_v \rceil$ new vertices in $V(G')$ where $b = \left\lceil \tfrac{2 n}{\epsilon \lambda} \right\rceil$. Also, $G'$ has an  edge $u_1 u_2$ iff $f(u_1)f(u_2) \in E(G)$, where $f: V(G') \rightarrow V(G)$ is the function mapping each vertex $u \in V(G')$ to its corresponding vertex in $V(G)$. (So $G'$ is a blow-up of $G$ by independent sets.)  We define a vertex partition on $G'$ by  $\mathcal V' = \{ f^{-1} (U) : U \in \mathcal V \}$ and a weight function $w'$ on $G'$ by $w'(u) = w(f(u))$. 
 
 Let us note a few bounds on $G'$. The weight of $G'$ is given by
$$ w'(G') = \sum\limits_{u \in V(G')} w'(u) = \sum\limits_{v \in V(G)} \lceil b \gamma_v \rceil w(v)  \geq b \sum\limits_{v \in V(G)} \gamma_v w(v). $$

Now consider some vertex $u \in V(G')$ with $f(u) = v \in V(G)$. Since $\gamma$ satisfies the LP, we have 
\begin{align*}
\deg(u) &=  \sum\limits_{x \in N(v)} \lceil b \gamma_x \rceil \leq \sum\limits_{x \in N(v)} (b \gamma_x + 1) \leq n + b \sum\limits_{x \in N(v)} \gamma_x \leq n + b \delta.
\end{align*}

Each block $U'  = f^{-1}(U) \in \mathcal V'$ has size $|U'| = \sum\limits_{x \in U} \lceil b \gamma_x \rceil$; since $\sum_{x \in U} \gamma_x = 1$, this implies that 
\begin{equation}
\label{usizeeqn}
b \leq |U'| \leq n + b.
\end{equation}

Next, in light of Eq.~(\ref{usizeeqn}), we form a graph $G''$ by discarding the lowest-weight $|U'| - b$ vertices in each block $U' \in \mathcal V'$. The resulting blocks are $b$-regular, and clearly $\Delta(G'') \leq \Delta(G') \leq n + b \delta$. Since $w$ is non-negative, discarding the lowest-weight vertices gives $w'(G'') \geq (\frac{b}{n+b}) w'(G')$.  

We apply Lemma~\ref{sparse-multi} to this graph $G''$ and weight function $w'$, with parameters $\epsilon' = \frac{\epsilon}{1-\epsilon}$ and $\lambda' = \lambda/2$  in place of $\epsilon, \lambda$. The preconditions of Lemma~\ref{sparse-multi} hold since since $w$, and hence $w'$, is non-negative, and for $\epsilon < 1/2$ we have:
\begin{align*}
\frac{b}{\Delta(G'')} &\geq \frac{b}{n +b \delta} = \frac{1}{n/b + \delta}  \geq \frac{1}{ \tfrac{\epsilon}{2} + (1/2 - \epsilon)} = 2 (1 + \epsilon').
\end{align*}

This gives an IT $M''$ of $G''$ with  $w'(M'') \geq \tfrac{w'(G'')}{b} (1 - \tfrac{\lambda}{2})$. Then $M = f(M'')$ is an IT of $G$ with weight $w(M) = w'(M'')$.  Since $n/b \leq \lambda/2$, we then have:
\begin{align*}
w(M) \geq  \frac{w'(G'')}{b} \Bigl( 1 - \frac{\lambda}{2} \Bigr) \geq  \Bigl( \frac{b}{n+b} \Bigr) \Bigl( 1 - \frac{\lambda}{2} \Bigr) w'(G) \geq (1 - \lambda)  \sum_{v \in V(G)} \gamma_v w(v) 
\end{align*}
as desired. For fixed $\delta$ and $\lambda$, the values $\epsilon', \lambda'$ are fixed, and the values $b, |V(G'')|$ are polynomial in $n$. So the expected runtime is $\poly(n)$.
\end{proof}

\begin{theorem}\label{ABZ-thm4}
There is a randomized algorithm which takes as input a parameter $\delta \in (0,1/2)$, a graph $G$ with vertex partition where $\mathcal P_{G, \delta} \neq \emptyset$, and a weight function $w$ on $G$, and returns an IT in $G$ with weight at least $\tau_{G, w,\delta}$. For fixed $\delta$, the expected runtime is $\poly(n)$.
\end{theorem}
\begin{proof}
Let  $\epsilon = 1/2 - \delta$.  We begin by sorting the vertices in each block in descending order of weight; the vertices in block $U$ are labeled as $v_{U,1}, v_{U,2}, \dots, v_{U,t_U}$ with $w(v_{U,1}) \geq w(v_{U,2}) \geq \dots \geq w(v_{U,t_U})$.  Next, we solve the LP to obtain a solution $\gamma \in \mathcal P_{G, \delta}$ with $\sum_{v \in V(G)} \gamma_v w(v) = \tau_{G, w, \delta}$. This takes $\poly(n)$ time since $\mathcal P_{G, \delta}$ has $\poly(n)$ constraints.  

To simplify the notation, we write $\gamma_{U,j}$ and  $w_{U,j}$ as short-hand for $\gamma_{v_{U,j}}$ and $w(v_{U,j})$, respectively. 

Since $\sum_{v \in U} \gamma_v = 1$, each block $U$ has a smallest index $s_U$ with $\sum_{k = 1}^{s_U} \gamma_{U,k} \geq 1 - \epsilon$. Form a new graph $G'$ by discarding vertices $v_{U,j}$ with $j > s_U$ in each block $U$, and define a weight function $w'$ on $G'$ by $w'(v) = w(v) - w_{U,s_U}$ for $v \in U$. Because of the sorted order of the vertices, $w'$ is non-negative. We also define a multiplicity vector $\gamma' \in [0,1]^{V(G')}$ for $G'$; again, to simplify the notation we write $\gamma'_{U,j}$ instead of $\gamma'_{v_{U,j}}$. The vector $\gamma'$ is defined by
$$
\gamma'_{U,j} = \begin{cases}
\frac{\gamma_{U,j}}{1-\epsilon} & \text{if $j < s_U$} \\
1 - \frac{\sum_{k=1}^{s_U-1} \gamma_{U,k}}{1-\epsilon} & \text{if $j = s_U$.}  
\end{cases}
$$

Note that $\gamma'_{U,s_U} \geq 0$ by definition of $s_U$. For $j < s_U$, we have $\gamma'_{U, j} = \frac{\gamma_{U,j}}{1-\epsilon} \leq \frac{\sum_{i=0}^j \gamma_{U,i}}{1-\epsilon} \leq \frac{1-\epsilon}{1-\epsilon} \leq 1$, and clearly $\gamma'_{U, s_U} \leq 1$. We also claim that the following bound holds for all $U, j$:
\begin{equation}
\label{ty5}
\gamma'_{U,j} \leq \frac{\gamma_{U,j}}{1-\epsilon}.
\end{equation}

It is clear for $j < s_U$, while for $j = s_U$, we have
\begin{align*}
\gamma'_{U,s_U} &= \frac{1 - \epsilon -  \sum_{k=1}^{s_U-1} \gamma_{U,k}}{1-\epsilon} =  \frac{\gamma_{U,s_U} + 1 - \epsilon  -  \sum_{k=1}^{s_U} \gamma_{U,k}}{1-\epsilon } \leq  \frac{\gamma_{U,s_U}}{1-\epsilon}.
\end{align*}

We next claim that $\gamma' \in \mathcal P_{G', \delta'}$  where $\delta' = \frac{1/2 - \epsilon}{1 - \epsilon} \in (0,1/2)$. To see this, note that the constraint $\sum_{k} \gamma'_{U,k} = 1$ follows from the definition of $\gamma'_{U,s_U}$ and the constraint $\sum_{u \in N(v)} \gamma'_u \leq \delta'$ follows from Eq.~(\ref{ty5}) and the fact that $\gamma \in \mathcal P_{G, \delta}$.

So let us set $\lambda' = \epsilon$ and apply Proposition~\ref{ABZ-thm3} to $G', w', \gamma', \delta', \lambda'$ to get an IT $M$ of $G'$ with 
$$
w'(M) \geq (1 - \lambda') \negthickspace \sum_{v \in V(G')} \negthickspace \gamma'_{v} w'(v) =  \frac{1 - \lambda'}{1 - \epsilon} \sum_{U \in \mathcal V} \sum_{j = 1}^{s_U - 1}  \gamma_{U,j} (w_{U,j} - w_{U,s_U}) =  \sum_{U \in \mathcal V} \sum_{j = 1}^{s_U - 1}  \gamma_{U,j} (w_{U,j} - w_{U,s_U}),
$$
where we omit the summand $j = s_U$ since $w'(v_{U,s_U}) = 0$. Since $M$ is an IT of $G'$, we have 
\begin{align*}
w(M) &= w'(M) + \sum_{U \in \mathcal V} w_{U,s_U}   \geq \sum_{U \in \mathcal V} \Bigl( w_{U,s_U} +  \sum_{j = 1}^{s_U - 1}  \gamma_{U,j} (w_{U,j} - w_{U,s_U}) \Bigr) \\
&= \sum_{U \in \mathcal V} \Bigl( w_{U,s_U} +  \sum_{j = 1}^{t_U}  \gamma_{U,j} (w_{U,j} - w_{U,s_U}) - \sum_{j=s_U}^{t_U} \gamma_{U,j} (w_{U,j} - w_{U,s_U}) \Bigr).
\end{align*}

Since $\sum_{j = 1}^{t_U} \gamma_{U,j} = 1$ and $\sum_U \sum_{j=1}^{t_U} \gamma_{U,j} w_{U,j} = \tau_{G, w, \delta}$, this is equal to $$
\tau_{G, w, \delta}-  \sum_{U \in \mathcal V}  \sum_{j=s_U}^{t_U} \gamma_{U,j} (w_{U,j} - w_{U,s_U}).
$$
Furthermore, since $w_{U,j} \leq w_{U,s_U}$ for $j > s_U$, this in turn is at least $\tau_{G, w, \delta}$.

When $\delta$ is fixed, then so are $\delta', \lambda'$, and so Proposition~\ref{ABZ-thm3} has $\poly(n)$ expected runtime.
\end{proof}

As a simple corollary, this gives  Theorem~\ref{main} (stated again here for convenience).
\begin{reptheorem}{main}
There is a randomized algorithm which takes as inputs a parameter $\epsilon > 0$, a graph $G$ with a $b$-regular vertex partition where $b \geq (2 + \epsilon) \Delta$, and a weight function $w: V(G) \rightarrow \mathbb R$, and finds an IT in $G$ with weight at least $w(G)/b$. For fixed $\epsilon$, the expected runtime is $\poly(n)$.
\end{reptheorem}
\begin{proof}
Let $\delta = \frac{1}{2 + \epsilon}$ and observe that $\gamma_{v} = \frac{1}{b}$ is a solution to $\mathcal P_{G, \delta}$, since every vertex $v$ has
$$
\sum_{u \in N(v)} \gamma_u = \sum_{u \in N(v)} \frac{1}{b} \leq \frac{\Delta}{b} \leq \frac{1}{2 + \epsilon}  = \delta,
$$
and clearly $\sum_{v \in U} \gamma_v =  1$ for each block $U$. We thus have $\tau_{G, w, \delta} \geq \sum_v \gamma_v w(v) = w(G)/b$.  Now apply Theorem~\ref{ABZ-thm4} with parameter $\delta$; note that if $\epsilon$ is fixed then so is $\delta$.
\end{proof}

\subsection{Weighted PITs}
\label{pit-sec}
An LP relaxation similar to $\mathcal P_{G, \delta}$ can be formulated for weighted PITs. Given a graph $G$, vertex partition $\mathcal V$, weight function $w: V(G) \rightarrow \mathbb R$, and value $\delta \in \mathbb R_{\geq 0}$, we define $\tau^*_{G, w, \delta}$ to be the largest objective function value to the following LP denoted $\mathcal P^*_{G, \delta}$.
\begin{equation*}
\begin{array}{rcrclcl}
  \displaystyle \max & \multicolumn{3}{l}{\sum\limits_{v\in V(G)}w(v)\gamma_v} \\
\textrm{subject to} & \sum\limits_{u\in N(v)}&\gamma_u & \le & \delta & & \forall v\in V(G)\\
&\sum\limits_{v\in U}&\gamma_v & \le &1 & & \forall U \in \mathcal V \\
& 0\le & \gamma_v & \le&1 & & \forall v \in V(G). 
\end{array}
\end{equation*}

Aharoni, Berger, and Ziv~\cite{Aharoni2007} proved the following:\footnote{The statement in~\cite{Aharoni2007} uses very different notation, and is formulated in terms of the dual LP.}
\begin{theorem}[{\cite[Theorem 10]{Aharoni2007}}]\label{ABZ}
For any weight function $w$, $G$ has a PIT of weight at least $\tau^*_{G, w,1/2}$.
\end{theorem}

 This is used by~\cite{Aharoni2007} to show Theorem~\ref{th2ex}. Our results for weighted ITs lead to the following analogue of Theorem~\ref{ABZ}:
\begin{corollary} 
There is a randomized algorithm which takes as inputs a parameter $\delta \in (0,1/2)$, a graph $G$ with a vertex partition  and a weight function $w$, and finds a PIT in $G$ with weight at least $\tau^*_{G, w,\delta}$. For fixed $\delta$, the expected runtime is $\poly(n)$.
\end{corollary}
\begin{proof}
Let $G'$ be the graph obtained by adding, for each block $U$, an isolated dummy vertex $x_U$ with weight zero. Any solution $\gamma$ to $\mathcal P^*_{G, \delta}$ corresponds a solution to $\mathcal P_{G', \delta}$, by setting $\gamma_{x_U} = 1 - \sum_{v \in U} \gamma_v$. Thus, applying Theorem~\ref{ABZ-thm4} to graph $G'$ gives an IT $M'$ with weight at least $\tau_{G', w, \delta} \geq \tau^*_{G, w, \delta}$.  Removing the dummy vertices from $M'$ yields a PIT $M$ of $G$ with $w(M) = w(M') \geq \tau^*_{G, w, \delta}$.
\end{proof}

\section{Independent Transversals with Vertex Restrictions}
\label{sec:restrict}
A number of combinatorial constructions use independent transversals with additional constraints. One common  restriction is that the IT must include certain vertices or be disjoint from a given set of vertices. Some LLL-based algorithms for ITs can accomodate these restrictions, sometimes with additional slack in parameters \cite{Harris2016a}.  

Our results on weighted independent transversals give the following crisp characterization:
\begin{theorem}
\label{avoid-thm}
There is a randomized algorithm which takes as inputs a parameter $\epsilon > 0$, a graph $G$ with a vertex partition $\mathcal V$ where $b^{\min} \geq (2 + \epsilon) \Delta$, and a vertex set $L \subseteq V(G)$ of size $|L| < b^{\text{min}}$. It returns an IT disjoint from $L$. For fixed $\epsilon$, the expected runtime is $\poly(n)$.
\end{theorem}
\begin{proof}
By discarding extra vertices from each block, we may assume without loss of generality that $\mathcal V$ is $b$-regular and $b \geq (2 + \epsilon) \Delta$ and $|L| < b$. Define a weight function on $V(G)$ by $w(v)=-1$ for $v \in L$ and $w(v)=0$ for $v\notin L$.  Now use Theorem~\ref{main} to obtain an IT $M$ with $w(M) \geq w(G)/b = -|L|/b > -1$. Since $w$ takes only values $-1$ and $0$, it must be that $w(M) = 0$ and hence $M \cap L = \emptyset$.
\end{proof}

As a simple corollary of Theorem~\ref{avoid-thm}, we can also get ITs which \emph{include} certain given vertices.
\begin{corollary}
\label{2v-avoid}
There is a randomized algorithm which takes as inputs a parameter $\epsilon > 0$, a graph $G$ with a vertex partition $\mathcal V$ where $b^{\text{min}}\geq (2+\epsilon) \Delta$, and a pair of vertices $v_1, v_2$ in the same block as each other. It returns a pair of ITs $M_1$ and $M_2$ of $G$ such that $v_1 \in M_1, v_2 \in M_2$  and $M_1 \setminus \{v_1\} = M_2 \setminus \{v_2 \}$. For fixed $\epsilon$, the expected runtime is $\poly(n)$.
\end{corollary}
\begin{proof}
Let $U = \mathcal V(v_1) = \mathcal V(v_2)$. Then apply Theorem~\ref{avoid-thm} to graph $G' = G[ V(G) \setminus U]$ with associated partition $\mathcal V' = \mathcal V \setminus \{ U \}$ and  with $L = N(v_1) \cup N(v_2)$. Here $|L| \leq 2 \Delta(G') < b^{\text{min}}(\mathcal V')$ as required. This generates an IT $M'$ of $G'$; now set $M_i = M' \cup \{v_i\}$ for $i = 1,2$.
\end{proof}

 By applying Corollary~\ref{2v-avoid} with $v_1 = v_2$, we get the following even simpler corollary:
\begin{corollary}
\label{1v-avoid}
There is a randomized algorithm which takes as inputs a parameter $\epsilon > 0$, a graph $G$ with a vertex partition where $b^{\text{min}} \geq (2+\epsilon) \Delta$, and a vertex $v \in V(G)$, and returns  an IT $M$ containing $v$. For fixed $\epsilon$, the expected runtime is $\poly(n)$.
\end{corollary}

As some examples, a construction in \cite{Lo2019} uses a (non-algorithmic) version of Corollary~\ref{2v-avoid}. We will also use Corollary~\ref{1v-avoid} in our application to strong colouring next.  This demonstrates the power of weighted independent transversals and Theorem~\ref{main}, even in contexts without an  overt weight function.

\section{Strong Colouring}\label{SC}
Aharoni, Berger, and Ziv~\cite{Aharoni2007} showed that the strong chromatic number is at most $3\Delta$ using an  extension of Theorem~\ref{th1ex} giving a sufficient condition for the existence of an IT containing a specified vertex.  Using Corollary~\ref{1v-avoid} for this instead, we obtain the following strong colouring algorithm:

\begin{corollary}\label{randstrong}
There is a randomized algorithm that takes as input a graph $G$ with a $b$-regular vertex partition $\mathcal V$ where $b \ge(3+\epsilon)\Delta$, and returns a strong $b$-colouring of $G$ with respect to $\mathcal V$. For fixed $\epsilon$, the expected runtime is $\poly(n)$.
\end{corollary}
\begin{proof}
The proof is essentially the same as that of~\cite{Aharoni2007}, so we just give a sketch. Consider a partial strong $b$-colouring $c$ of $G$ with respect to $\mathcal V$, an uncoloured vertex $v$, and a colour $\alpha$ not used by $c$ on the block $\mathcal V(v)$. Define a new graph $G'$ by removing from each block $U$ the vertices whose colour appears on the neighbourhood of the vertex $x_U$ in $U$ coloured $\alpha$ (if it exists). This reduces the size of each block by at
most $\Delta(G)$. Then we apply Corollary~\ref{1v-avoid} to find an IT $Y$ of $G'$ containing $v$. As shown in~\cite{Aharoni2007}, if we modify $c$ by giving each vertex $y\in Y$ colour $\alpha$ and the corresponding vertex $x_{\mathcal V(y)}$ colour $c(y)$, we obtain a partial strong $b$-colouring with strictly fewer uncoloured vertices than $c$ (in particular it colours $v$). Hence in at most $n$ such steps we get a strong $b$-colouring of $G$.
\end{proof}
Analogous to the connection between chromatic number and fractional chromatic number, there is a fractional version of strong colouring for a graph $G$ with a $b$-regular vertex partition. By LP duality, this has two equivalent definitions:
\begin{itemize}
\item (Primal) For all weight functions $w: V(G) \rightarrow \mathbb R$, there is an IT $M$ of $G$ with $w(M) \geq w(G)/b$.
\item (Dual)  There exists a function $f$ mapping each IT $M$ of $G$ to a real number $f(M) \in [0,1]$ such that $\sum_M f(M) = b$ and for all vertices $v \in V(G)$ it holds $\sum_{M \ni v} f(M) = 1$.
\end{itemize}

Observe that if the function $f$ in the dual definition takes values in the range $f(M) \in \{0, 1 \}$, then $f$ is a strong 
colouring with respect to the vertex partition. The fractional version of the strong colouring conjecture mentioned in Section~\ref{intro} was shown by \cite{Aharoni2007}:
\begin{theorem}[\cite{Aharoni2007}]\label{ABZfrac}
Every graph $G$ is fractionally strongly $2\Delta(G)$-colourable.
\end{theorem}

Theorem~\ref{ABZfrac}, in terms of the primal definition of fractional strong colouring, is simply a restatement of Theorem~\ref{th2ex}. Theorem~\ref{main} can be viewed as an algorithmic counterpart. There is also a generic method of \cite{carr-vempala} to convert this into an algorithmic version of the dual definition.  We quote the following crisp formulation of \cite{anegg}:
 \begin{theorem}[\cite{anegg}]
\label{prim-dual-thm}
 Suppose that $\mathfrak S$ is a collection of subsets of ground set $U$ with associated weights $g_u$ for each $u \in U$. Suppose that there is a polynomial-time algorithm which takes as input a weight function $w: U \rightarrow \mathbb R_{\geq 0}$, and returns some $S \in \mathfrak S$ with $\sum_{u \in S} w(u) \geq \sum_{u \in U} w(u) g_u$. 
 
 Then there is a polynomial-time algorithm to generate a subcollection $\mathfrak S_0 \subseteq \mathfrak S$ with $|\mathfrak S_0| \leq |U|$,  with associated weights $\lambda: \mathfrak S_0 \rightarrow [0,1]$, such that $\sum_{S \in \mathfrak S_0} \lambda(S) = 1$ and for every $u \in U$, it holds that $\sum_{S \in \mathfrak S_0: u \in S} \lambda(S) \geq g_u$.
 \end{theorem}
 
 As an immediate corollary of Theorem~\ref{prim-dual-thm} and Theorem~\ref{main}, we obtain the following:
  \begin{theorem}
\label{prim-dual-thm2}
There is a randomized algorithm which takes as input a parameter $\epsilon > 0$ and a graph $G$ with $b$-regular vertex partition $\mathcal V$ where $b \geq (2 + \epsilon) \Delta$, and finds ITs $I_1, \dots, I_n$ with associated weights $f_1, \dots, f_n \geq 0$, such that $f_1 + \dots + f_n = b$ and for every vertex $v \in V$ it holds that $\sum_{i: v \in I_i} f_i = 1$. For fixed $\epsilon$, the expected runtime is $\poly(n)$.
\end{theorem}
\begin{proof}
Apply Theorem~\ref{prim-dual-thm} where $\mathfrak S$ is the collection of ITs of $G$, and where the ground set $U$ is $V(G)$, and where $g_v = 1/b$ for every vertex $v$. By Theorem~\ref{main}, we have a polynomial-time procedure to find an IT $S$ for any given weight function $w$ with $\sum_{v \in S} w(v) \geq \sum_{v \in V} w(v)/b$. 
By Theorem~\ref{prim-dual-thm}, we get a collection $\mathfrak S_0 = \{I_1, \dots, I_n \}$ of ITs and weights $\lambda$ with $\sum_{I \in \mathfrak S_0: v \in I} \lambda(I) \geq 1/b$ for all $v$ and $\sum_{I \in \mathfrak S_0} \lambda(I) = 1$. We set $f_i  = b \lambda(I_i)$ for each $i$. Since each $I \in \mathfrak S_0$ has exactly one vertex in each block, we must have $\sum_{i: v \in I_i} f_i = 1$ for all $v$. 
\end{proof}

\section{Acknowledgments}
Thanks to journal and conference reviewers for many helpful corrections and suggestions.

\end{document}